%% file: main.tex
\documentclass[a4paper,USenglish,authorcolumns]{lipics-v2021}

\usepackage{cite}   
\usepackage[utf8]{inputenc}
\usepackage{multicol}
\usepackage[linesnumbered,ruled,vlined]{algorithm2e}
\SetAlFnt{\ttfamily}
\SetArgSty{ttfamily}
\usepackage{textcomp}

\usepackage{relsize}
\usepackage{hyperref}
\usepackage{listings}

\def\code{\texttt}

\def\dali{{Dal\'{\i}}}
\def\sync{\code{sync}}
\def\lincas{\code{lin\_CAS}}
\def\pnew{\code{pnew}}
\def\pretire{\code{pretire}}
\def\pdetach{\code{pdetach}}
\def\pdelete{\code{pdelete}}
\def\beginop{\code{begin\_op}}
\def\endop{\code{end\_op}}
\def\abortop{\code{abort\_op}}
\def\resetop{\code{reset\_op}}

\theoremstyle{plain}

\definecolor{code_bg}{rgb}{0.98,0.98,0.98}
\definecolor{code_hl}{rgb}{1.0,0.98,0.8}
\def\lb{\linebreak[1]}
\newcommand*{\hightlightCode}[1]{\setlength{\fboxsep}{0pt}\colorbox{code_hl}{#1}}
\makeatletter

\def\lb{\linebreak[1]}

\newif\ifverbose \verbosetrue
\newif\iffullver \fullvertrue

\title{Fast Nonblocking Persistence for Concurrent Data Structures}

\author{Wentao Cai\footnotemark}{University of
Rochester, NY, USA}{wcai6@ur.rochester.edu}{}{}

\author{Haosen Wen\footnotemark[\value{footnote}]}{University of
Rochester, USA}{hwen5@ur.rochester.edu}{}{}

\author{Vladimir Maksimovski}{University of
Rochester, USA}{vmaksimo@ur.rochester.edu}{}{}

\author{Mingzhe Du}{University of
Rochester, USA}{mdu5@ur.rochester.edu}{}{}

\author{Rafaello Sanna}{University of
Rochester, USA}{rsanna@ur.rochester.edu}{}{}

\author{Shreif Abdallah}{University of
Rochester, USA}{selsaid@ur.rochester.edu}{}{}

\author{Michael L. Scott}{University of Rochester, USA}{scott@ur.rochester.edu}{}{}

\authorrunning{W. Cai, H. Wen, V. Maksimovski, M. Du, R. Sanna, S.
Abdallah, and M. L. Scott} 

\Copyright{W. Cai, H. Wen, V. Maksimovski, M. Du,
R. Sanna, S. Abdallah, and M. L. Scott} 



\category{} 

\iffullver
\else
\relatedversiondetails{Full Version}{https://arxiv.org/abs/2105.09508}
\fi
\supplementdetails{Source Code}{https://github.com/urcs-sync/Montage} 

\funding{This work was supported in part by NSF grants CCF-1717712, CNS-1900803, CNS-1955498, and by a Google Faculty Research award.}

\EventEditors{Seth Gilbert}
\EventNoEds{1}
\EventLongTitle{35th International Symposium on Distributed Computing (DISC 2021)}
\EventShortTitle{DISC 2021}
\EventAcronym{DISC}
\EventYear{2021}
\EventDate{October 4--8, 2021}
\EventLocation{Freiburg, Germany (Virtual Conference)}
\EventLogo{}
\SeriesVolume{209}
\ArticleNo{2}

\hypersetup{
  colorlinks,
  linkcolor=teal, 
  citecolor=magenta, 
  urlcolor=blue, 
}

\begin{document}

\maketitle

\footnotetext[\value{footnote}]{The first two authors (Wentao Cai and Haosen Wen) contributed equally to this paper.}



\begin{abstract}
We present a fully lock-free variant of our recent Montage system for
persistent data structures. The variant, nbMontage, adds persistence
to almost any nonblocking concurrent structure without introducing
significant overhead or blocking of any kind. Like its predecessor,
nbMontage is \emph{buffered durably linearizable}: it guarantees that
the state recovered in the wake of a crash will represent a consistent
prefix of pre-crash execution. Unlike its predecessor, nbMontage
ensures wait-free progress of the persistence frontier, thereby
bounding the number of recent updates that may be lost on a crash, and
allowing a thread to force an update of the frontier (i.e., to perform
a \code{sync} operation) without the risk of blocking. As an extra
benefit, the helping mechanism employed by our wait-free \sync\
significantly reduces its latency.

Performance results for nonblocking queues, skip lists, trees, and
hash tables rival custom data structures in
the literature---dramatically faster than achieved with prior
general-purpose systems, and generally within 50\% of equivalent
non-persistent structures placed in DRAM.
\end{abstract}


\input{intro.tex}
\input{related.tex}
\input{impl.tex}
\input{proof.tex}
\input{exp.tex}
\input{conclusions.tex}

\bibliographystyle{abbrv}
\bibliography{main}

\appendix
\input{appendix.tex}

\end{document}

%% file: intro.tex
\section{Introduction}
\label{sec:intro}

With the advent of dense, inexpensive nonvolatile memory (NVM), it is
now feasible to retain pointer-based, in-memory data structures across
program runs and even hardware reboots.  So long as caches remain
transient, however, programs must be instrumented with write-back and
fence instructions to ensure that such structures remain consistent in
the wake of a crash.  Minimizing the cost of this instrumentation
remains a significant challenge.

Section~\ref{sec:related} of this paper surveys both bespoke
persistent data structures and general-purpose systems. Most of these
structures and systems ensure that execution is \emph{durably
linearizable}~\cite{izraelevitz-disc-2016}: not only will every
operation linearize (appear to occur atomically) sometime between its
call and return; each will also \emph{persist} (reach a state in which
it will survive a crash) before returning, and the persistence order
will match the linearization order.  To ensure this agreement between
linearization and persistence, most existing work relies on locks to
update the data structure and an undo or redo log together,
atomically.

Unfortunately, strict durable linearizability forces expensive
write-back and fence instructions onto the critical path
of every operation.  A few data structures---e.g., the
\dali\ hash table of Nawab et al.~\cite{nawab-disc-2017} and the InCLL
MassTree of Cohen et al.~\cite{cohen-asplos-2019}---reduce the cost of
instrumentation by supporting only a relaxed, \emph{buffered}
variant~\cite{izraelevitz-disc-2016} of durable linearizability.
Nawab et al.\ dub the implementation technique \emph{periodic
persistence}. In the wake of post-crash recovery, a buffered durably
linearizable structure is guaranteed to reflect some prefix of the
pre-crash linearization order---but with the possibility that recent
updates may be lost.

Independently, some data structures---notably, the B-tree variants of
Yang et al.~\cite{yang-fast-2015}, Oukid et
al.~\cite{oukid-sigmod-2016}, and Chen et al.~\cite{chen-vldb-2020},
and the hash set of Zuriel et al.~\cite{zuriel-oopsla-2019}---reduce
the overhead of persistence by writing back and fencing only the data
required to rebuild a semantically equivalent structure during
post-crash recovery. A mapping, for example, need only persist a pile
of key-value pairs. Index structures, which serve only to increase
performance, can be kept in faster, volatile memory.

Inspired by \dali\ and the InCLL MassTree, our group introduced the
first \emph{general purpose} system~\cite{wen-icpp-2021} for
buffered durably linearizable data structures.  Like the
sets of the previous paragraph, \emph{Montage}
facilitates persisting only semantically essential
\emph{payloads}.  Specifically, it employs a global \emph{epoch
clock}.  It tracks the semantically essential updates performed in
each operation and ensures that (1) no operation appears to span an
epoch boundary, and (2) no essential update fails to persist for two
consecutive epochs.  If a crash occurs in epoch $e$, Montage recovers
the abstract state of the structure from the end of epoch $e-2$.

In return for allowing a bit of work to be lost on a crash, and for
being willing to rebuild, during recovery, any portion of a structure
needed only for good performance, users of Montage can expect to
achieve performance limited primarily by the read and write latency of
the underlying memory---not by instrumentation overhead.  When a
strict persistence guarantee is required (e.g., before printing a
completion notice on the screen or responding to a client over the
network), Montage allows an application thread to invoke an explicit
\sync\ operation---just as it would for data in any conventional file
or database system.

Unfortunately, like most general-purpose persistence systems,
Montage relies on locks at critical points in the code.  All of the
reported performance results are for lock-based data structures.  More
significantly, while the system can support nonblocking structures,
progress of the persistence frontier is fundamentally blocking.
Specifically, before advancing the epoch from $e$ to $e+1$, the system
waits for all operations active in epoch $e-1$ to complete. If one of
those operations is stalled (e.g., due to preemption), epoch advance
can be indefinitely delayed.  This fact implies that the \sync\
operation cannot be nonblocking.  It also implies that the size of the
history suffix that may be lost on a crash cannot be bounded \emph{a
priori}.

While blocking is acceptable in many circumstances, nonblocking
structures have compelling advantages.
\iffullver
They are immune to deadlocks
caused by lock misordering or priority inversion, and to performance
anomalies and high tail latency due to inopportune preemption.
Lock-free and wait-free nonblocking structures~\cite[Sec.\
3.7]{herlihy-shavit-book} are also immune to livelock. By assigning a
coherent abstract state to every reachable concrete state, nonblocking
data structures simplify recovery~\cite{cai-ismm-2020} and eliminate
the need to perform logging for the sake of failure
atomicity~\cite{david-atc-2018}. As long envisioned in the theory
community, they may even facilitate data sharing among processes that
fail independently~\cite{hedayati-atc-2019-hodor}.

\fi
In the original paper on durable linearizability, Izraelevitz et
al.~\cite{izraelevitz-disc-2016} presented a mechanical transform that
will turn any linearizable nonblocking concurrent data structure into
an equivalent persistent structure.  Unfortunately, this transform is
quite expensive, especially for highly concurrent programs: it litters
the code with write-back and fence instructions.  More efficient
strategies for several specific data structures, including
queues~\cite{friedman-ppopp-2018} and hash tables~\cite{chen-atc-2020,
zuriel-oopsla-2019}, have also been developed by hand.

Friedman et al.~\cite{friedman-pldi-2020} observed that many
nonblocking operations begin with a read-only ``traversal'' phase in
which instrumentation can, with care, be safely elided; applying this
observation to the transform of Izraelevitz et al.\ leads to
substantially better performance in many cases, but the coding process
is no longer entirely mechanical. Ramalhete et
al.~\cite{ramalhete-dsn-2019} and Beadle et
al.~\cite{beadle-qstm-2020} present nonblocking persistent software
transactional memory (STM) systems, but both have fundamental serial
bottlenecks that limit scalability.

In this paper, we extend Montage to produce the first
\emph{general-purpose, high-performance system for nonblocking
periodic persistence}. Our \emph{nbMontage} allows programmers, in a
straightforward way, to convert most wait-free or lock-free
linearizable concurrent data structures into fast, equivalent
structures that are lock-free and buffered durably linearizable.
(Obstruction-free structures can also be converted; they remain
obstruction free.) Like its predecessor, nbMontage requires every
nonblocking update operation to linearize at a compare-and-swap (CAS)
instruction that is guaranteed, prior to its execution, to constitute
the linearization point if it succeeds. (Read-only operations may
linearize at a \code{load}.)  Most nonblocking structures in the
literature meet these requirements.

Unlike its predecessor, nbMontage provides a wait-free \sync. Where
the original system required all operations in epoch $e-1$ to complete
before the epoch could advance to $e+1$, nbMontage allows
pending operations to remain in limbo: in the wake of a crash or of
continued crash-free execution, an operation still pending at the end
of $e-1$ may or may not be seen, sometime later, to have linearized in
that epoch. In addition, where the original Montage required
operations to accommodate ``spurious'' CAS failures caused by epoch
advance, nbMontage retries such CASes internally, without compromising
lock freedom. These changes were highly nontrivial:
they required new mechanisms to register (announce) pending
updates; distinguish, in recovery, between registered and linearized
updates; indefinitely delay (and reason about) the point at which an
update is known to have linearized; and avoid work at \sync\ time for
threads that have nothing to persist.

We have applied nbMontage to Michael \& Scott's
queue~\cite{michael-podc-1996}, Natarajan \& Mittal's binary search
tree~\cite{natarajan-ppopp-2014}, the rotating skip list of Dick et
al.~\cite{dick-rotating-2016}, Michael's hash
table~\cite{michael-spaa-2002}, and Shalev \& Shavit's extensible hash
table~\cite{shalev-acm-2006}. All conversions were
straightforward---typically less than 30 lines of code. By persisting
only essential data and avoiding writes-back and fences on the
critical path, nbMontage structures often approach or outperform not
only their equivalents in the original (blocking) Montage but also
their transient equivalents when data is placed (without any
algorithmic changes) in slower, nonvolatile memory.

Summarizing contributions:
\begin{enumerate}
\item
    We introduce nbMontage, the first general system for
    nonblocking periodic persistence.
\item
    We tailor the nbMontage API to nonblocking data
    structures.  With this new API, conversion of existing nonblocking
    structures for persistence is straightforward.
\item
    We argue that nbMontage provides buffered durable linearizability and
    wait-free \sync\ for compatible data structures, while still
    preserving safety and liveness.
\item
    We compare the performance of nbMontage structures both to their
    original, transient versions and to hand-crafted alternatives from the
    literature, running on a recent Intel server with Optane NVM\@. Our
    results confirm exceptional throughput, responsive \sync, negligible
    overhead relative to the original Montage, and reasonable recovery
    latency.
\end{enumerate}

%% file: related.tex
\section{Related Work}
\label{sec:related}

The past decade has seen extensive work on persistent data structures,
much of it focused on B-tree indices for file systems and databases~\cite{%
venkataraman-fast-2011, chen-vldb-2015, yang-fast-2015,
oukid-sigmod-2016, kim-tos-2018, hwang-fast-2018, liu-icpp-2019,
chen-vldb-2020, cha-tos-2020}. Other work has targeted
queues~\cite{friedman-ppopp-2018}, RB trees~\cite{wang-tos-2018},
radix trees~\cite{lee-fast-2017}, and hash
tables~\cite{zuriel-oopsla-2019, schwalb-imdm-2015,
nawab-disc-2017,chen-atc-2020}.

In recent years, \emph{durable linearizability} has emerged as the
standard correctness criterion for such
structures~\cite{izraelevitz-disc-2016, friedman-ppopp-2018,
zuriel-oopsla-2019, memaripour-asplos-2020}. This criterion builds on
the familiar notion of linearizability~\cite{herlihy-toplas-1990} for
concurrent (non-persistent) structures.  A structure is said to be
linearizable if, whenever threads perform operations concurrently, the
effect is as if each operation took place, atomically, at some point
between its call and return, yielding a history consistent with the
sequential semantics of the abstraction represented by the structure.
A persistent structure is said to be \emph{durably linearizable} if
(1) it is linearizable during crash-free execution, (2) each operation
persists (reaches a state that will survive a crash) at some point
between its call and return, and (3) the order of persists matches the
linearization order.

In addition to custom data structures, several groups have developed
general-purpose systems to ensure the failure atomicity of lock-based
critical sections~\cite{hsu-eurosys-2017, chakrabarti-oopsla-2014,
izraelevitz-asplos-2016, liu-micro-2018} or spec\-u\-la\-tive
transactions~\cite{volos-asplos-2011, coburn-asplos-2011,
charzistergiou-vldb-2015, giles-msst-2015, correia-spaa-2018,
cohen-oopsla-2018, memaripour-iccd-2018, ramalhete-dsn-2019,
gu-atc-2019, beadle-qstm-2020, pavlovic-podc-2018, intel-pmdk,
liu-asplos-2017, genc-pldi-2020}. Like the bespoke structures
mentioned above, all of these systems are durably linearizable---they
ensure that an operation has persisted before returning to the calling
thread.

As noted in Section~\ref{sec:intro}, nonblocking persistent structures
can achieve failure atomicity without the need for logging, since
every reachable concrete state corresponds to a well-defined abstract
state. At the same time, while a lock-based operation can easily
arrange to linearize and persist in the same order (simply by holding
onto the locks needed by any successor operation), a nonblocking
operation becomes visible to other threads the moment it linearizes.
As a result, those other threads must generally take care to ensure
that anything they read has persisted before they act upon it. Writing
back and fencing read locations is a major source of overhead in the
mechanical transform of Izraelevitz et
al.~\cite{izraelevitz-disc-2016}. Friedman et
al.~\cite{friedman-pldi-2020} avoid this overhead during the initial
``traversal'' phase of certain nonblocking algorithms. David et
al.~\cite{david-atc-2018} avoid redundant writes-back and fences in
linked data structures by \emph{marking} each updated pointer in one
of its low-order bits.  A thread that persists the pointer can use a
CAS to clear the bit, informing other threads that they no longer need
to do so.

In both blocking and nonblocking structures, the overhead of
persistence can be reduced by observing that not all data are
semantically meaningful.  In any implementation of a set or mapping,
for example, the items or key-value pairs must be persistent, but the
index structure can (with some effort) be \emph{rebuilt} during
recovery. Several groups have designed persistent B-trees, lists, or
hash tables based on this observation~\cite{yang-fast-2015,
oukid-sigmod-2016, chen-vldb-2020, mahapatra-2019}.

Ultimately, however, any data structure or general-purpose system
meeting the strict definition of durable linearizability will
inevitably incur the overhead of at least one persistence fence in
every operation~\cite{cohen-lowerbound-2018}. For highly concurrent
persistent structures, this overhead can easily double the latency of
every operation.
Similar observations, of course, have applied to I/O
operations since the dawn of electronic computing, and are routinely
addressed by completing I/O in the background.  For data structures
held in NVM, \emph{periodic persistence}~\cite{nawab-disc-2017}
provides an analogous ``backgrounding'' strategy, allowing a structure
to meet the more relaxed requirements of \emph{buffered} durable
linearizability---specifically, the state recovered after a crash is
guaranteed to represent a prefix of the linearization order of
pre-crash execution.

Nawab et al.~\cite{nawab-disc-2017} present a lock-based hash table,
\dali, that performs updates only by pre-pending to per-bucket lists,
thereby creating a revision history (deletions are effected by
inserting ``anti-nodes'').  A clever ``pointer-rotation'' strategy
records, for each bucket, the head nodes of the list for each of the
past few values of a global \emph{epoch clock}.  At the end of each
coarse-grain epoch, the entire cache is flushed.  In the wake of a
crash, threads ignore nodes prepended to hash bucket lists during the
two most recent epochs. No other writes-back or fences are required.
Cohen et al.~\cite{cohen-asplos-2019} also flush the cache at global
epoch boundaries, but employ a clever system of
\emph{in-cache-line-logging} (InCLL) to retain the epoch number and
beginning-of-epoch value for every modified leaf-level pointer in a
variant of the Masstree structure~\cite{mao-eurosys-2012}. In the wake
of a crash, threads use the beginning-of-epoch value for any pointer
that was modified in the epoch of the crash.

Both \dali\ and InCLL employ techniques that might be applicable in
principle to other data structures.  To the best of our knowledge,
however, Montage~\cite{wen-icpp-2021} is the only existing
general-purpose system for buffered durable linearizable structures.
It also has the advantage of persisting only semantically essential
data.  As presented, unfortunately, it provides only limited support
for nonblocking data structures, and its attempts to advance the epoch
clock can be arbitrarily delayed by stalled worker threads. Our
nbMontage resolves these limitations, allowing us to provide a
wait-free \sync\ operation and to bound the amount of work that may be
lost on a crash. It also provides a substantially simpler API\@.

%% file: impl.tex
\section{System Design}

\subsection{The Original Montage}\label{sec:montage} 

Semantically essential data
in Montage resides in \emph{payload} blocks in NVM\@.  Other data may
reside in transient memory.  The original system
API~\cite{wen-icpp-2021} is designed primarily for lock-based data
structures, but also includes support for nonblocking operations (with
blocking advance of the persistence frontier).  The typical operation
is bracketed by calls to \beginop\ and \endop.  In between, reads and
writes of payloads use special \emph{accessor} (\code{get} and
\code{set}) methods.

Internally, Montage maintains a slow-ticking \emph{epoch clock}.  In
the wake of a crash in epoch $e$, the Montage recovery procedure
identifies all and only the payloads in existence at the end of epoch
$e-2$.  It provides these, through a parallel iterator, to the
restarted application, which can then rebuild any needed transient
structures. Accessor methods allow payloads that were created in the
current epoch to be modified in place, but they introduce significant
complexity to the programming model.

Payloads are created and deleted with \pnew\ and \pdelete.  These
routines are built on a modified version of our
Ralloc~\cite{cai-ismm-2020} persistent allocator.  In the original
Ralloc, a tracing garbage collector was used in the wake of a crash to
identify and reclaim free blocks.  In the Montage version, epoch
\emph{tags} in payloads are used to identify all and only those blocks
created and not subsequently deleted at least two epochs in the past.
To allow a payload to survive if a crash happens less than two epochs
after a deletion, deletions are implemented by creating
\emph{anti-nodes}.  These are automatically reclaimed, along with their
corresponding payloads, after two epochs have passed.

To support nonblocking operations, the original Montage provides a
\code{CAS\_verify} routine that succeeds only if it can do so in the
same epoch as the preceding \beginop.  If \code{CAS\_verify} fails, the
operation will typically start over; before doing so, it should call
\abortop\ instead of \endop, allowing the system to clean up without
persisting the operation's updates.

Regardless of persistence, nodes removed from a nonblocking structure
typically require some sort of \emph{safe memory reclamation}
(SMR)---e.g., epoch-based reclamation~\cite{fraser-thesis-2004} or
hazard pointers~\cite{michael-tpds-2004}---to delay the reuse of memory
until one can be sure that no thread still holds a transient reference.
In support of SMR, the original Montage provides a \pretire\ routine
that creates an anti-node to mark a payload as semantically deleted,
and no new operation can reference this payload.
In the absence of crashes, Montage will automatically reclaim the
payload and its anti-node 2--3 epochs after SMR calls \pdelete. In the
event of a crash, however, if two epochs have elapsed since \pretire, the
existence of the anti-node allows the Montage recovery procedure to
avoid a memory leak by reclaiming the retired payload even in the
absence of \pdelete. This is safe since the epoch of the \pretire\ 
is persisted, and all operations with references to this payload are
gone after the crash.
Significantly, since a still-active operation will
(in the original Montage) prevent the epoch from advancing far enough to
persist anything in its epoch, \pretire\ can safely be called after the
operation has linearized, so long as it has not yet called \endop.

When a program needs to ensure that operation(s) have persisted (e.g.,
before printing a confirmation on the screen or responding to a client
over the network), Montage allows it to invoke an explicit \sync.  The
implementation simply advances the epoch from $e$ to $e+2$ (waiting
for operations in epochs $e-1$ and $e$ to complete).
The two-epoch convention
avoids the need for quiescence~\cite{nawab-disc-2017}: a
thread can advance the epoch from $e$ to $e+1$ while other threads are
still completing operations that will linearize in $e$.

The key to buffered durable linearizability is to ensure that every
operation that updates payloads linearizes in the epoch with which
those payloads are tagged.  Each epoch boundary then captures a prefix
of the data structure's linearization order. Maintaining this
consistency is straightforward for lock-based operations.  In the
nonblocking case, \code{CAS\_verify} uses a variant of the
double-compare-single-swap (DCSS) construction of Harris
al.~\cite{harris-disc-2002} 
\iffullver 
(App.~\ref{app:dcss}) 
\else
(see App.\ B in \href{https://arxiv.org/abs/2105.09508}{Full
Version} for complete pseudocode) 
\fi
to confirm the
expected epoch and perform a conditional update, atomically.
Unfortunately, the fact that an epoch advance from $e$ to $e+1$ must
wait for operations in $e-1$ means that even if a data structure
remains nonblocking during crash-free execution, the progress of
persistence itself is fundamentally blocking.  This in turn implies
that calls to \sync\ are blocking, and that it is not possible,
\emph{a priori}, to bound the amount of work that may be lost on a
crash.

\lstdefinestyle{myCustomCppStyle}{
  language=C++,
  stepnumber=1,
  tabsize=1,
  showspaces=false,
  showstringspaces=false,
  xleftmargin=0pt,
  basicstyle=\scriptsize\ttfamily\selectfont,
  numberstyle=\ttfamily\tiny,
  escapeinside=||,
  keywordstyle=\bfseries\color{red!40!black},
  commentstyle=\itshape\color{green!40!black},
  columns=fullflexible,
}

\begin{figure}[!t]
\begin{lstlisting}[style=myCustomCppStyle]  
class PBlk; // Base class of all Payload classes
class Recoverable { // Base class of all persistent objects
  template <class payload_type> payload_type* pnew(...); // Create a payload block
  void pdelete(PBlk*); // Delete a payload; should be called only when safe, e.g., by SMR.
  void pdetach(PBlk*); // Mark payload for retirement if operation succeeds
  void sync(); // Persist all operations that happened before this call
  vector<PBlk*>* recover(); // Recover and return all recovered payloads
  void abort_op(); // Optional routine to abandon current operation
};
template <class T> class CASObj { // Atomic CAS object that provides load and CAS
  /* Epoch-verifying linearization method: */
  bool lin_CAS(T expected, T desired) {
    begin_op(); // write back or delete old payloads as necessary; tag new ones
    while (1) { // iterations can be limited for liveness
      ... // main body of DCSS
      switch (DCSS_status) {
        case COMMITTED: end_op(); return true; // clean up metadata
        case FAILED: abort_op(); return false; // untag payloads and clear retirements; clean up metadata
        case EPOCH_ADVANCED: reset_op(); // update and reuse payloads and retirements
      }
    }
  /* Non-verifying atomic methods: */
  T load();  void store(T desired);  bool CAS(T expected, T desired);
};
\end{lstlisting}
  \vspace{-3ex}
  \captionsetup{justification=centering}
  \caption{C++ API of nbMontage, largely inherited from the original Montage. }
  \vspace{-3ex}
  \label{fig:api}
\end{figure}

Note, however, that since
\code{CAS\_verify} will succeed only in the expected epoch, any
nonblocking operation that lags behind an epoch advance is doomed to
fail and retry in a subsequent epoch. There is no need to wait for it
to resume, explicitly fail, and unregister from the old epoch in
\abortop.  Unfortunately, the waiting mechanism is deeply embedded in
the original Montage implementation---e.g., in the implementation of
\pretire, as noted above. Overall, there are four nontrivial issues
that must be addressed to build nbMontage:
\begin{description}
  \item[(Sec.~\ref{sec:pblk})] {\label{issue:pblk}
    Every operation must register its pending updates 
    (both new and to-be-deleted payloads)
    before reaching its linearization point, so that an
    epoch advance can help it to persist even if it stalls immediately
    after the linearization point.
  }
  \item[(Sec.~\ref{sec:rec})] {\label{issue:rec}
    The recovery procedure must be able to distinguish an epoch's
    ``real'' payloads and anti-nodes from those that may have been
    registered for an operation that failed due to a CAS conflict or
    epoch advance.
  }
  \item[(Sec.~\ref{sec:buf})] {\label{issue:buf}
    The buffering containers that record persistent blocks to be
    written back or deleted need a redesign, in order to accommodate an
    arbitrary number of epochs in which operations have not yet noticed
    that they are doomed to fail and retry in a new epoch.
  }
  \item[(Sec.~\ref{sec:adv})] {\label{issue:adv}
    An epoch advance must be able to find and persist any blocks
    (payloads or anti-nodes) that were created in the previous epoch but
    have not yet been written back and fenced.  If \sync\ is to be fast,
    this search must avoid iterating over all active threads.
  }
\end{description}

\subsection{nbMontage API}\label{sec:api}

As shown in Figure~\ref{fig:api}, the nbMontage API reflects three
major changes.  First, because the epoch can now advance from $e$ to
$e+1$ even when an operation is still active in $e-1$, we must
consider the possibility that a thread may remove a payload from the
structure, linearize, and then stall.  If a crash occurs two epochs
later, we must ensure that the removed payload is deleted during
post-crash recovery, to avoid a memory leak: post-linearization
\pretire\ no longer suffices.  Our nbMontage therefore replaces
\pretire\ with a \pdetach\ routine that must be used to register
to-be-deleted payloads \emph{prior} to linearization.  As in the
original Montage, deletion during crash-free execution is the
responsibility of the application's SMR\@.

Second, again because of nonblocking epoch advance, nbMontage requires
that payloads visible to more than one thread be treated as immutable.
Updates always entail the creation of new payloads; \code{get} and
\code{set} accessors are eliminated.

Third, in a dramatic simplification, nbMontage replaces the
original \beginop, \endop, and \code{CAS\_verify} with a new \lincas\
(linearizing CAS) routine. (The \abortop\ routine is also rolled into
\lincas, but remains in the API so operations can call it explicitly
if they choose, for logical reasons, to start over.) When \lincas\ is
called, all payloads created by the calling thread since its last
completed operation (and not subsequently deleted) will be tagged with
the current epoch, $e$.  All anti-nodes stemming from \pdetach\ calls
made since the last completed operation (and not corresponding to
payloads created in that interval) will likewise be tagged with $e$.
The \lincas\ will then attempt a DCSS and, if the current epoch is
still $e$ and the specified location contains the expected value, the
operation will linearize, perform the internal cleanup previously
associated with \endop, and return \code{true}. If the DCSS fails
because of a conflicting update, \lincas\ will perform the internal
cleanup associated with \abortop\ and return \code{false}. If the DCSS
fails due to epoch advance, \lincas\ will update the operation's
payloads and anti-nodes to the new epoch and retry.  By ensuring,
internally, that the epoch never advances unless some thread has
completed an operation 
\iffullver
(App.~\ref{app:progress}),
\else
(App.\ C in \href{https://arxiv.org/abs/2105.09508}{Full
Version}),
\fi
nbMontage can ensure
that some thread makes progress in each iteration of the retry loop.

Programmers using nbMontage are expected to obey the following
constraints:
\begin{enumerate}
\item\label{constraint:ori-lin}
    Each nbMontage data structure $R$ must be designed to be nonblocking
    and linearizable during crash-free execution when nbMontage is
    disabled. Specifically, $R$ must be linearizable when (a)
    \pnew\ and \pdelete\ are ordinary \code{new} and \code{delete}; (b)
    \pdetach\ and \sync\ are no-ops; and (c) \code{CASObj} is \code{std::atomic}
    and \lincas\ is ordinary \code{CAS}.
\item\label{constraint:lincas}
    Every update operation of $R$ linearizes in a pre-identified
    \lincas---one that is guaranteed, before the call, to comprise the
    operation's linearization point if it succeeds.  Any
    operation that conflicts with or depends upon this update must access
    the \lincas's target location. Read-only operations may linearize at a
    \code{load}.
\item\label{constraint:immutability}
    Once attached to a structure (made visible to other threads), a
    payload is immutable.  Changes to a structure are made only by
    adding and removing payloads.
\item\label{constraint:payloads}
    The semantics of each operation are fully determined by the set of
    payloads returned by \pnew\ and/or passed to \pdetach\ prior
    to \lincas.
\newcounter{lastconstraintA}
\setcounter{lastconstraintA}{\theenumi}
\end{enumerate}

Pseudocode for nbMontage's core classes and methods appears in
Figure~\ref{fig:pseudo}; these are discussed and referred to by
pseudocode line numbers in the following subsections.
Appendix~\ref{app:mmm_hash} presents the
changes required to port Maged Michael's lock-free hash
table~\cite{michael-spaa-2002} to nbMontage.

\subsection{Updates to Payloads}\label{sec:pblk}

To allow the epoch clock to advance without blocking, nbMontage
abandons in-place update of payloads.  It interprets \pdetach\ as
requesting the creation of an \emph{anti-node}. An anti-node shares an
otherwise unique, hidden ID with the payload being detached. Newly created
payloads and anti-nodes are buffered until the next \lincas\ in their
thread.  If the \lincas\ succeeds, the buffered nodes will be visible
to epoch advance operations, and will persist even if the creating
thread has stalled.

In the pseudocode of Figure~\ref{fig:pseudo}, calls to \pnew\ and
\pdetach\ are held in the \code{allocs} and \code{detaches}
containers. Anti-nodes are created, and both payloads and anti-nodes
are tagged, in \beginop\
(lines~\ref{code:begin-detaches-begin}--\ref{code:begin-allocs-end},
called from within \lincas). If the \lincas\ fails due to conflict,
\abortop\ resets \pnew-ed payloads so they can be reused in the
operation's next attempt
(lines~\ref{code:abort-allocs-begin}--\ref{code:abort-allocs-end}); it
also withdraws \pdetach\ requests, allowing the application to detach
something different the next time around
(lines~\ref{code:abort-detaches-begin}--\ref{code:abort-detaches-end}).
If attempts occur in a loop (as is common), the programmer may call
\pnew\ outside the loop and \pdetach\ inside, as shown in
Figure~\ref{fig:example} (App.~\ref{app:mmm_hash}). 
If an operation no
longer needs its \pnew-ed payloads (e.g., after a failed insertion),
it may call \pdelete\ to delete them; this automatically erases them
from \code{allocs} (line~\ref{code:pdelete-erase-allocs}). The
internal \resetop\ routine serves to update and reuse both payloads
and anti-nodes in anticipation of retrying a DCSS that fails due to
epoch advance
(lines~\ref{code:reset-op-begin}--\ref{code:reset-op-end}).

\SetKwRepeat{Do}{do}{while}
\SetKw{Or}{or}
\SetKw{And}{and}
\SetKw{Return}{return}
\SetKw{Delete}{delete}
\SetKw{Throw}{throw}
\SetKw{Sfence}{sfence}
\SetKw{Local}{\texttt{thread\_local}}
\SetKwProg{Fn}{Function}{}{}
\SetKwProg{Macro}{Macro}{}{}
\SetKwProg{Class}{Class}{}{}
\SetKwProg{Struct}{Struct}{}{}

\begin{figure*}[!ht]
  \setlength{\columnsep}{0pt}
  \begin{algorithm}[H]
  \scriptsize
  \DontPrintSemicolon 
  \begin{multicols}{2}
  \Struct{\textup{\texttt{CircBuffer}}}{\label{code:cb-begin}
    uint64 cap\;
    atomic<uint64> pushed,popped\;
    PBlk* blks[cap]\;
    \Fn{\textup{\texttt{push(PBlk* item)}}}{\label{code:push-begin}
      cpush=pushed.load()\;
      cpop=popped.load()\;
      \If{\textup{\texttt{cpush$\geq$cpop+cap}}}{
        clwb(blks[cpop\%cap])\;
        popped.CAS(cpop,cpop+1)\;
      }
      blks[cpush\%cap]=item\;
      pushed.store(cpush+1)\tcp*[l]{single producer}
    }\label{code:push-end}
    \Fn{\textup{\texttt{pop\_all()}}}{\label{code:pop-all-begin}
      cpop=popped.load()\;\label{code:pop-all-retry}
      cpush=pushed.load()\;
      \If{\textup{\texttt{cpop==cpush}}}{
        return\;
      }
      \ForEach{\textup{\texttt{i from cpop to cpush}}}{
        break if i\%cap reaches cpop\%cap twice\;
        clwb(blks[i\%cap])\;
      }
      popped.CAS(cpop,cpush)\;
    }\label{code:pop-all-end}
  }\label{code:cb-end}
  atomic<uint64> global\_epoch\;
  Mindicator mindi\;
  \Local uint64 e\_curr,e\_last\;
  \Local vector<PBlk*> allocs,detaches\;
  \Local int tid\;
  int thd\_cnt\tcp*[l]{number of threads}
  CircBuffer TBP[thd\_cnt][4]\;\label{code:tbp}
  \Struct{\textup{\texttt{Desc}}}{\label{code:desc-begin}
    uint64 old,new,epoch=0\;
    uint64 type=DESC\; \label{code:desc-type}
    uint64 tid\_sn=0\;\label{code:desc-tid-sn}
    \tcp{64 bits for ref to CASObj \\64 for cnt, with last 2 for status}
    atomic<uint128> r\_c\_s\;
    \Fn{\textup{\texttt{reinit(uint64 e)}}}{
      tid\_sn++\;
      r\_c\_s.store(0)\;
      epoch=e\;
    }
  }\label{code:desc-end}
  \Struct{\textup{\texttt{PBlk}}}{
    (void* vtable)\;
    PBlk* anti=NULL\;
    uint64 epoch=0\;
    uint64 type=\{PAYLOAD,ANTI\}\;\label{code:pblk-type}
    uint64 tid\_sn=0,blk\_uid=0\;\label{code:pblk-tid-sn}
    \Fn{\textup{\texttt{setup(uint64 t,Desc* desc,uint64 e)}}}{\label{code:pblk-setup-begin}
      type=t\;
      epoch=desc?0:desc.epoch\;\label{code:setup-epoch}
      tid\_sn=desc?0:desc.tid\_sn\;\label{code:setup-sn}
      TBP[tid][e\%4].push(this)\;
    }
    \Fn{\textup{\texttt{destructor()}}}{\label{code:pblk-destructor-begin}
      epoch=0\;
      clwb(this)
    }\label{code:pblk-destructor-end}
  }
  Desc descs[thd\_cnt]\;
  vector<PBlk*> TBF[thd\_cnt][4]\;\label{code:tbf}
  \Fn{\textup{\texttt{pdelete(PBlk* pblk)}}}{\label{code:pdelete-begin}
    e=global\_epoch.load()\;
    allocs.erase(pblk)\tcp*[l]{no-op if pblk not in allocs}\label{code:pdelete-erase-allocs}
    TBF[tid][(e+1)\%4].insert(pblk$\rightarrow$anti)\;
    TBF[tid][e\%4].insert(pblk)\;
  }\label{code:pdelete-end}
  \Fn{\textup{\texttt{begin\_op(bool reset=false)}}}{
    e\_curr=global\_epoch.load()\;\label{code:load-epoch}
    \If{\textup{\texttt{e\_last<e\_curr}}}{\label{code:begin-tbp-begin}
      TBP[tid][e\_last\%4].pop\_all()\;\label{persist-beginop}
      clwb(descs[t])\;\label{code:beginop-desc-flush}
      update t's val to e\_curr in mindi\;\label{code:beginop-update-advance}
    }
    descs[tid].reinit(e\_curr)\;\label{code:desc-reinit}
    \For{\textup{\texttt{i from e\_last-1 to min(e\_last+1,e\_curr-2)}}}{\label{code:begin-tbf-begin}
      delete all items in TBF[tid][i\%4]\;\label{code:tbf-pop-all}
      sfence\;\label{code:begin-tbf-end}
    }
    \ForEach{\textup{\texttt{r in detaches}}}{\label{code:begin-detaches-begin}
      \If(\tcp*[h]{default branch}){\textup{\texttt{!reset}}}{
        allocate an anti-node anti for r\;
        r$\rightarrow$anti=anti\;
        r$\rightarrow$anti$\rightarrow$blk\_uid=r$\rightarrow$blk\_uid\;
      }
      r$\rightarrow$anti$\rightarrow$setup(ANTI,descs[tid],e\_curr)\;
    }
    \ForEach{\textup{\texttt{p in allocs}}}{
      p$\rightarrow$setup(PAYLOAD,descs[tid],e\_curr)\;
    }\label{code:begin-allocs-end}
  }
  \Fn{\textup{\texttt{end\_op()}}}{
    detaches.clear()\;
    allocs.clear()\;
    e\_last=e\_curr\;
    e\_curr=0\;
  }
  \Fn{\textup{\texttt{abort\_op(bool reset=false)}}}{
    \If(\tcp*[h]{to reuse anti-nodes}){\textup{\texttt{reset}}}{
      \ForEach{\textup{\texttt{r in detaches}}}{\label{code:reset-detaches-begin}
        r$\rightarrow$anti$\rightarrow$setup(ANTI,NULL,e\_curr)\;
      }
    }\Else(\tcp*[h]{default branch:delete anti-nodes}){
      \ForEach{\textup{\texttt{r in detaches}}}{\label{code:abort-detaches-begin}
        delete(r$\rightarrow$anti)\;\label{code:anti-reset}
        r$\rightarrow$anti=NULL\;
      }
      detaches.clear()\;\label{code:abort-detaches-end}
    }
    \ForEach{\textup{\texttt{p in allocs}}}{\label{code:abort-allocs-begin}
      p$\rightarrow$setup(PAYLOAD,NULL,e\_curr)\;\label{code:p-reset}
    }\label{code:abort-allocs-end}
    e\_last=e\_curr\;
    e\_curr=0\;
  }
  \Fn{\textup{\texttt{reset\_op()}}}{\label{code:reset-op-begin}
    abort\_op(true)\;
    begin\_op(true)\;
  }\label{code:reset-op-end}
  \Fn{\textup{\texttt{advance()}}}{\label{code:advance-begin}
    e=global\_epoch.load()\;
    \ForEach{\textup{\texttt{t in mindi whose val==e-1}}}{\label{code:advance-loop}
      TBP[t][(e-1)\%4].pop\_all()\;\label{code:persist-advance}
      clwb(descs[t])\;\label{code:advance-desc-flush}
      update t's val to e in mindi\;\label{code:advance-update-advance}
    }
    sfence\;
    \If{some op linearized in e-1 or e}{\label{code:increment-epoch-condition}
      global\_epoch.CAS(e,e+1)\;\label{code:increment-epoch}
    }
    \vspace{-1ex}
  }\label{code:advance-end}
  \end{multicols}
  \vspace{-5ex}
  \end{algorithm}
  \caption{nbMontage pseudocode.}
  \label{fig:pseudo}
  \vspace{-3ex}
\end{figure*}

\subsection{CAS and Recovery}\label{sec:rec}

The implementation of \lincas\ employs an array of persistent
\emph{descriptors}, one per thread.  These form the basis of the DCSS
construction~\cite{harris-disc-2002} 
\iffullver
(App.~\ref{app:dcss}). 
\else
(App.\ B in \href{https://arxiv.org/abs/2105.09508}{Full Version}).
\fi
Each descriptor records CAS parameters (the old value, new value, and
CAS object); the epoch in which to linearize; and the status of the
CAS itself (in progress, committed, or
failed---lines~\ref{code:desc-begin}--\ref{code:desc-end}). After a
crash, the recovery procedure must be able to tell when a block in NVM
(a payload or anti-node) appears to be old enough to persist, but
corresponds to an operation that did not commit. Toward that end, each
block contains a 64-bit field that encodes the thread ID and a
monotonic serial number; together, these constitute a unique
\emph{operation ID} (lines~\ref{code:desc-tid-sn}
and~\ref{code:pblk-tid-sn}). At the beginning of each operation
attempt, \beginop\ updates the descriptor, incrementing its serial
number (line~\ref{code:desc-reinit}).  Previous uses of the descriptor
with smaller serial numbers are regarded as having committed; blocks
corresponding to those versions remain valid unless they are deleted
or reinitialized (lines~\ref{code:tbf-pop-all},~\ref{code:anti-reset},
and~\ref{code:p-reset}). Deleting or reinitializing a persistent block
resets its epoch to zero and registers it to be written back in the
current epoch
(lines~\ref{code:pblk-setup-begin}--\ref{code:pblk-destructor-end}).
Registration ensures that resets persist, in \beginop, \emph{before the
next update to the descriptor}
(lines~\ref{code:begin-tbp-begin}--\ref{code:desc-reinit}). During an
epoch advance from $e$ to $e+1$, the descriptors of operations in
$e-1$ are written back (at line~\ref{code:advance-desc-flush}) to
ensure that their statuses reach NVM before the update of the global
epoch clock.

Informally, an nbMontage payload is said to be \emph{in use} if it has
been created and not yet detached by linearized operations.
Identifying such payloads precisely is made difficult by the existence
of \emph{pending} operations---those that have started but not yet
completed, and whose effects may not yet have been seen by other
threads. In the wake of a crash in epoch $e$, nbMontage runs through
the Ralloc heap, assembling a set of potentially allocated blocks and
finding all CAS descriptors (identified by their \code{type}
fields---lines~\ref{code:desc-type} and~\ref{code:pblk-type}). By
matching the serial numbers and thread IDs of blocks and descriptors,
the nbMontage-internal recovery procedure identifies all and only the
payloads that are known, as of the crash, to have been in use at the
end of epoch $e-2$. Specifically, if block $B$ has thread ID $t$,
serial number $s$, and epoch tag $f\!$, nbMontage will recover $B$ if
and only if
\begin{enumerate}
  \item $0 < f \le e-2$\label{blk-rec:2};
  \item $(s < \mathtt{descs}[t].\mathtt{sn}) \lor
  (s = \mathtt{descs}[t].\mathtt{sn} \land \mathtt{descs}[t].\mathtt{status} =
      \mathtt{COMMITTED})$; and\label{blk-rec:5}
  \item if $B$ is a payload, it has not been canceled by an in-use
      anti-node.\label{blk-rec:6}
\end{enumerate}

Once the in-use blocks have been identified, nbMontage returns them to
a data-structure-specific recovery routine that rebuilds any needed
transient state, after which the state of the structure is guaranteed
to reflect some valid linearization of pre-crash execution through the
end of epoch $e-2$.

\subsection{Buffering Containers}\label{sec:buf} 

Persistent blocks created or deleted in a given epoch will be recorded
in thread- and epoch-specific to-be-persisted (TBP) and to-be-freed
(TBF) containers. Every thread maintains four statically allocated
instances of each (only 3 are actually needed, but indexing is faster
with 4---Fig.~\ref{fig:pseudo}, lines~\ref{code:tbp}
and~\ref{code:tbf}).

TBPs are fixed-size circular buffers.  When a buffer is full, its
thread removes and writes back a block before inserting a new one.  In
the original version of Montage, epoch advance always occurs in a
dedicated background thread (the \sync\ operation handshakes with this
thread). As part of the advance from epoch $e$ to $e+1$, the
background thread iterates over all worker threads $t$, waits for $t$
to finish any active operation in $e-1$, extracts all blocks from
TBP$[t][(e-1)\!\!\!\mod 4]$, and writes those blocks back to memory.

Insertions and removals from a TBP buffer never occur concurrently in
the original version of Montage.  In nbMontage, however, an operation
that is lagging behind may not yet realize that it is doomed to retry,
and may still be inserting items into the buffer when another thread
(or several other threads!) decide to advance the epoch. The lagging
thread, moreover, may even be active in epoch $e-1-4k$, for some $k>0$
(lines~\ref{code:anti-reset} and~\ref{code:p-reset}). This concurrency
implies that TBPs need to support single-producer-multiple-consumers
(SPMC) concurrency. Our implementation of the SPMC buffer
(lines~\ref{code:cb-begin}--\ref{code:cb-end}) maintains two monotonic
counters, \code{pushed} and \code{popped}. To insert an item, a thread
uses \code{store} to increment \code{pushed}. To remove some item(s),
a thread uses \code{CAS} to increase \code{popped}. For simplicity,
the code exploits the fact that duplicate writes-back are semantically
harmless: concurrent removing threads may iterate over overlapping
ranges (lines~\ref{code:pop-all-begin}--\ref{code:pop-all-end}).

TBFs are dynamic-size, thread-unsafe containers implemented as vectors
(line~\ref{code:tbf}). 
Although deletion
must respect epoch ordering, it can be performed lazily during
crash-free execution, with each thread responsible for the blocks in
its own TBFs. In \beginop, \emph{after it has updated its descriptor},
thread $t$ deletes blocks in \mbox{TBF$[t][i\!\!\mod 4]$}, for $i \in
[e_\mathit{last}-1, \min(e_\mathit{last}+1,e_\mathit{curr}-2)]$, where
$e_\mathit{last}$ is the epoch of $t$'s last operation and
$e_\mathit{curr}$ is the epoch of its current operation
(lines~\ref{code:begin-tbf-begin}--\ref{code:begin-tbf-end}).

\subsection{Epoch Advance}\label{sec:adv} 

To make \sync\ nonblocking, we first decentralize the original epoch
advance in Montage so that instead of making a request of some
dedicated thread, every thread is now able to advance the epoch on its
own. In the worst case, such an epoch advance may require iterating
over the TBP buffers of all threads in the system.  In typical cases,
however, many of those buffers may be empty.  To reduce overhead in
the average case, we deploy a variant of Liu et al.'s
\emph{mindicator}~\cite{liu-icdcs-2013} to track the oldest epoch in
which any thread may still have an active operation. Implemented as a
wait-free, fixed-size balanced tree, our variant represents each
thread and its current epoch as a leaf.  An ancestor in the tree
indicates the minimum epoch of all nodes in its subtree.  When thread
$t$ wishes to advance the epoch from $e$ to $e+1$, it first checks to
see whether the root of the mindicator is less than $e$.  If so, it
scans up the tree from its own leaf until it finds an ancestor with
epoch $<e$.  It then reverses course, traces down the tree to find a
lagging thread, persists its descriptor and any blocks in the
requisite TBP container, and repeats until the root is at least $e$.  
When multiple threads call \sync\ concurrently, this
nearest-common-ancestor strategy allows the work of persistence to be
effectively parallelized.  Experiments described in
Section~\ref{sec:sync-exp} confirm that our use of the mindicator,
together with the lazy handling of TBF buffers
(Section~\ref{sec:buf}), leads to average \sync\ times on the order of a
few microseconds.

%% file: proof.tex
\renewcommand{\proofname}{Proof (sketch)}

\section{Correctness}

We argue that nbMontage preserves the linearizability and lock freedom
of a structure implemented on top of it,
and adds buffered durable linearizability.  We also argue that advances
of the persistence frontier in nbMontage are wait free.

\subsection{Linearizability}

\begin{theorem}\label{theorem:lin}
Suppose that $R$ is a data structure obeying the constraints of
Section~\ref{sec:api}, running on nbMontage, and that $K$ is realizable
concrete history of $R$.  $K$ is linearizable.
\end{theorem}

\begin{proof}
The \pnew, \pdelete, and \lincas\ routines of nbMontage have
the same semantics as \code{new}, \code{delete}, and \code{CAS}
calls in the original data structure.  The \code{pdetach} routine has no
semantic impact on crash-free execution: it simply ensures that a block
whose removal has linearized will be reclaimed in post-crash recovery.
The \code{sync} routine, similarly, has no semantic impact---with no
parameters and no return values, it can linearize anywhere.
If the instructions comprising each call to \pnew, \pdelete, \pdetach,
\sync, and \lincas\ in a concrete nbMontage history are replaced with
those of
\code{new}, \code{delete}, \code{no-op}, \code{no-op}, and \code{CAS},
respectively, the result will be a realizable concrete history of the
original data structure.  Since that history is linearizable, so is the
one on nbMontage.
\end{proof}

\subsection{Buffered Durable Linearizability}
\label{sec:proof-BDL}

As is conventional, we assume that each concurrent data structure
implements some abstract data type.  The semantics of such a type are
defined in terms of legal \emph{abstract sequential
  histories}---sequences of operations (request-response pairs), with
their arguments and return values.  We can define the \emph{abstract
  state} of a data type, after some finite sequential history $S$, as
the set of possible extensions to $S$ permitted by the type's semantics.
In a \emph{concurrent} abstract history $H$, invocations and responses
may be separated, and some responses may be missing, in which case the
invocation is said to be \emph{pending} at end of $H.$
$H$ is said to be \emph{linearizable} if (1) there exists a history $H'$
obtained by dropping some subset of the pending invocations in $H$ and
adding matching responses for the others, and (2) there exists a
sequential history $S$ that is equivalent to $H'$ (same invocations and
responses) and that respects both the \emph{real-time} order of $H'$ and
the semantics of the abstract type.  $S$ is said to be a
\emph{linearization of} $H.$

Suppose now that $R$ is a linearizable nonblocking implementation of
type $T,$ and that $r$ is a concrete state of $R$---the bits in memory
at the end of some concrete (instruction-by-instruction) history $K$\@.
For $R$ to be correct there must exist a mapping
$\mathcal{M}$ such that for any such $K$ and $r,$ $\mathcal{M}(r)$ is
the abstract state that results from performing, in order, the
abstract operations corresponding to concrete operations that have
linearized in $K$.

A structure $R$ is \emph{buffered durably linearizable} if post-crash
recovery always results in some concrete state $s$ that is justified by
some prefix $P$ of pre-crash concrete execution---that is, there exists
a linearization $S$ of the abstract history corresponding to $P$ such
that $\mathcal{M}(s)$ is the abstract state produced by $S$.

Consider again the \thelastconstraintA\ constraints listed at the end of
Section~\ref{sec:api} for data structures running on nbMontage.
Elaborating on constraints~\ref{constraint:immutability}
and~\ref{constraint:payloads}, we use $r|_p$ to denote the set of
payloads that were created (and inserted) and not yet detached by the 
operations that generated $r$.  This allows us to recast
constraint~\ref{constraint:payloads} and to add an additional
constraint:
\begin{enumerate}
\setcounter{enumi}{\thelastconstraintA}
\item[\textbf{\thelastconstraintA\rlap{$'$}.}]\label{constraint:pstate-mapping}
    There exists a mapping $\mathcal{Q}$ from sets of payloads to states
    of $T$ such that for any concrete state $r$ of $R$,
    $\mathcal{M}(r) = \mathcal{Q}(r|_p)$.
\item\label{constraint:recover-with-pstate}
    The recovery procedure of $R$, given a set of in-use payloads $p$,
    constructs a concrete state~$s$ such that $\mathcal{M}(s) =
    \mathcal{Q}(p)$.
\addtocounter{enumi}{-\value{lastconstraintA}}
\addtocounter{enumi}{-1}
\refstepcounter{enumi}
\label{extraconstraints}
\end{enumerate}

\begin{theorem}\label{theorem:BDL} 
If a crash happens in epoch $e$, $R$ will recover to a concrete
state $s$ such that $\mathcal{M}(s)$ is the abstract
state produced by some linearization $S$ of the abstract history $H$
comprising pre-crash execution through the end of epoch $e-2$.
In other words, $R$ is buffered durably linearizable.
\end{theorem}

\begin{proof}
For purposes of this proof, it is convenient to say that an update
operation that commits the descriptor of its \lincas\
linearizes on the \emph{preceding load} of the global epoch
clock---the one that double-checks the clock before commit. Under this
interpretation, if $r$ is the concrete state of memory at the end of
epoch $e-2$, we can say that $\mathcal{M}(r)$ reflects a sequential
history containing all and only those operations that have committed
their descriptors 
\iffullver
(line~\ref{code:committing-CAS} in
Fig.~\ref{fig:dcss} of App.~\ref{app:dcss})
\else
(line 17 in Fig.\ 8 in
\href{https://arxiv.org/abs/2105.09508}{Full Version})
\fi
by the end of the epoch.
But this is not the only possible linearization of execution to that
point!  In particular, any operation that has loaded
\code{global\_epoch} (line~\ref{code:load-epoch} in
Fig.~\ref{fig:pseudo}), initialized its descriptor
(line~\ref{code:desc-reinit} of Fig.~\ref{fig:pseudo}), and installed
that descriptor in a \code{CASObj} 
\iffullver
(line~\ref{code:install-desc} of
Fig.~\ref{fig:dcss}) 
\else
(line 71 in
Fig.\ 8 in \href{https://arxiv.org/abs/2105.09508}{Full
Version}) 
\fi
but has not yet committed the descriptor may
``linearize in the past'' (i.e., in epoch $e-2$) if it \emph{or
another, helping operation} commits the descriptor in the future.  When a
crash occurs in epoch $e$, any such retroactively linearized
operations will see their payloads and anti-nodes included in the
state $s$ recovered from the crash. $\mathcal{M}(s)$ will therefore
correspond, by constraint~\ref{constraint:recover-with-pstate}, to the
linearization of execution through the end of epoch $e-2$ that
includes all and only those pending operations that have linearized by
the time of the crash. Crucially, if operation $b$ depends on
operation $a$, in the sense that $a$ has completed in any extension of
$H$ in which $b$ has completed, then, by
constraint~\ref{constraint:lincas} of Section~\ref{sec:api}, the
helping mechanism embodied by \lincas\ ensures that if $b$'s payloads
and anti-nodes are included in $s$, $a$'s are included also.
%
%
\end{proof}


\subsection{Wait-free Persistence}
\begin{theorem}\label{theorem:wait-free-advance}
    The epoch advance in nbMontage is wait free.
\end{theorem}

\begin{proof}
As shown in Fig.~\ref{fig:pseudo}, an epoch advance from $e$ to $e+1$
repeatedly finds a thread $t$ that may still be active in $e-1$
(line~\ref{code:advance-loop}), persists the contents of its TBP
container and its descriptor
(lines~\ref{code:persist-advance}--\ref{code:advance-desc-flush}), and
updates its mindicator entry.
In the worst case, identifying all threads with blocks to be persisted
requires time $O(T)$, where $T$ is the number of threads, since the
total size of the mindicator is roughly $2T$ nodes.
Since each TBP container has bounded size, all the data of a thread
can be persisted in $O(1)$ time.
Mindicator updates, worst case, take $O(T \log T)$ time.

Since \sync\ advances the epoch at most twice, it, too,
is wait free.
\end{proof}

\subsection{Lock freedom}
\begin{theorem}
    nbMontage preserves lock freedom during crash-free execution.
\end{theorem}

\begin{proof}
Given Theorem~\ref{theorem:wait-free-advance}, the only additional loop
introduced by nbMontage is the automatic retry that occurs inside
\lincas\ when the epoch has advanced.
While this loop precludes wait freedom, we can 
\iffullver
(with a bit of effort---see App.~\ref{app:progress})
\else
(with a bit of effort---see App.\ C in \href{https://arxiv.org/abs/2105.09508}{Full Version})
\fi
arrange to advance the epoch from
$e$ to $e+1$ only if some update operation has
linearized in epoch $e-1$ or $e$ 
(line~\ref{code:increment-epoch-condition} in Fig.~\ref{fig:pseudo}).
This suffices to preserve lock freedom.
As a corollary, a data structure that is obstruction free remains so
when persisted with nbMontage.
\end{proof}

%% file: exp.tex
\section{Experimental Results}

To confirm the performance benefits of buffering, we constructed
nbMontage variants of Michael \& Scott's
queue~\cite{michael-podc-1996}, Natarajan \& Mittal's binary search
tree~\cite{natarajan-ppopp-2014}, the rotating skip list of Dick et
al.~\cite{dick-rotating-2016}, Michael's chained hash
table~\cite{michael-spaa-2002}, and the resizable hash table of Shalev
\& Shavit~\cite{shalev-acm-2006}. Mappings keep their key-value pairs
in payloads and their index structures in transient memory. The queue
uses payloads to hold values and their order.  Specifically, between
its two loads of the queue tail pointer, the enqueue operation calls
\code{fetch\_add} on a global counter to obtain a serial number for
the to-be-inserted value. We benchmarked those data structures and
various competitors on several different workloads. Below are the
structures and systems we tested:
\begin{itemize}
\item \textbf{Montage} and \textbf{nbMontage} -- as described in
previous sections.
\item \textbf{Friedman} -- the persistent lock-free queue of Friedman
et al.~\cite{friedman-ppopp-2018}.
\item \textbf{Izraelevitz} and \textbf{NVTraverse} -- the N\&M tree,
the rotating skip list, and Michael's hash table persisted using
Izraelevitz's transform~\cite{izraelevitz-disc-2016} and the NVTraverse
transform~\cite{friedman-pldi-2020}, respectively. 
\item \textbf{SOFT} and \textbf{NVMSOFT} -- the lock-free hash table
of Zuriel et al.~\cite{zuriel-oopsla-2019}, which persists only
semantic data. SOFT keeps a full copy in DRAM, while
NVMSOFT is modified to keep and access values only in NVM.  Neither
supports \code{update} on existing keys.
\item \textbf{CLevel} -- The persistent lock-free hash table of Chen
et al.~\cite{chen-atc-2020}.
\item \textbf{Dal\'{\i}} -- the lock-based buffered durably
linearizable hash table of Nawab et al.~\cite{nawab-disc-2017}.
\item \textbf{DRAM\,(T)} and \textbf{NVM\,(T)} -- as a baseline for
comparison, these are unmodified transient versions of our various
data structures, with data located in DRAM and NVM, respectively.
\end{itemize}


\subsection{Configurations}

We configured Montage and nbMontage with 64-entry TBP buffers and an
epoch length of 10\,ms.  In practice, throughput is broadly
insensitive to TBP size, and remains steady with epochs as short as
100\,\textmu s. All experiments were conducted on an Intel Xeon Gold
6230 processor with 20 physical cores and 40 hyperthreads, six 128\,GB
Optane Series 100 DIMMs, and six 32\,GB DRAMs, running Fedora 30
Kernel 5.3.7 Linux Server. Threads are placed first on separate
physical cores and then on hyperthreads.  NVM is configured through
the \code{dax-ext4} file system in ``App Direct'' mode.

All experiments use JEMalloc~\cite{evans2006jemalloc} for transient
allocation and Ralloc~\cite{cai-ismm-2020} for persistent allocation,
with the exception of CLevel, which requires the allocator from
Intel's PMDK~\cite{intel-pmdk}. All chained hash tables have 1 million
buckets. The
warm-up phase for mappings inserts 0.5\,M key-value pairs drawn from a
key space of 1\,M keys.  Queues are initialized with 2000 items.
Unless otherwise specified, keys and values are strings of 32 and 1024
bytes, respectively. We report the average of three trials, each of
which runs for 30 seconds. Source code for nbMontage and the
experiments is available at \url{https://github.com/urcs-sync/Montage}.

\subsection{Throughput}\label{sec:throughput} 

Results for queues, binary search trees, skip lists, and hash tables
appear in Figures~\ref{fig:throughput-qts}--\ref{fig:throughput-ht}.
The nbMontage versions of the M\&S queue, N\&M tree, rotating skip
list, and Michael hash table outperform most persistent alternatives
by a significant margin---up to 2$\times$ faster than Friedman et
al.'s queue, 1.3--4$\times$ faster than NVTraverse and Izraelevitz et
al.'s transform, and 3--14$\times$ faster than CLevel and Dal\'{\i}.
Significantly, nbMontage achieves almost the same throughput as
Montage. SOFT and NVMSOFT are the only exceptions: the former benefits
from keeping a copy of its data in DRAM; both benefit from clever
internal optimizations.  The DRAM copy has the drawback of forgoing
the extra capacity of NVM; the optimization has the drawback of
precluding atomic \code{update} of existing keys. While the transient
Shalev \& Shavit (S\&S) hash table (DRAM\,(T)-SS in
Fig.~\ref{fig:throughput-ht}) is significantly slower than the
transient version of Michael's hash table (DRAM\,(T)), the throughput
of the Montage version (nbMontage-SS) is within 65\% of the transient
version and still faster than all other pre-existing persistent
options other than SOFT and NVMSOFT.

\begin{figure*}
    \strut\hfill
    \subfloat[Queues---1:1
    push:pop\label{fig:queues_thread}]%
        {\includegraphics[width=.3\textwidth]%
            {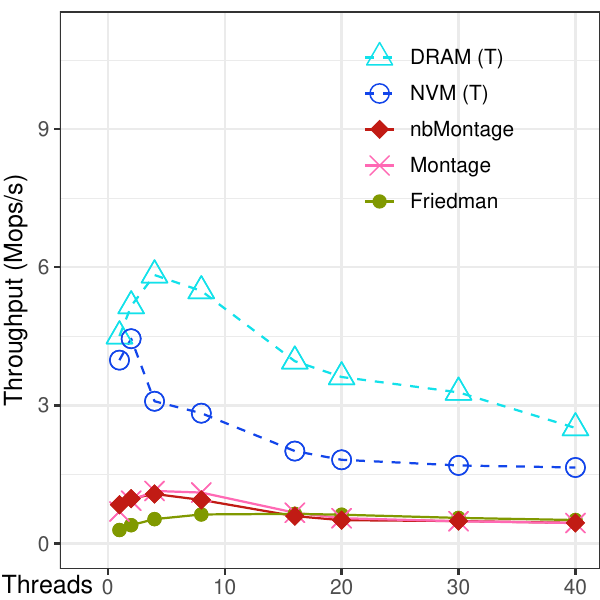}%
        \vspace{-1ex}%
        }
    \hfill
    \subfloat[Binary search trees---2:1:1
    get:insert:remove\label{fig:trees_g50i25r25_thread}]%
        {\includegraphics[width=.3\textwidth]%
            {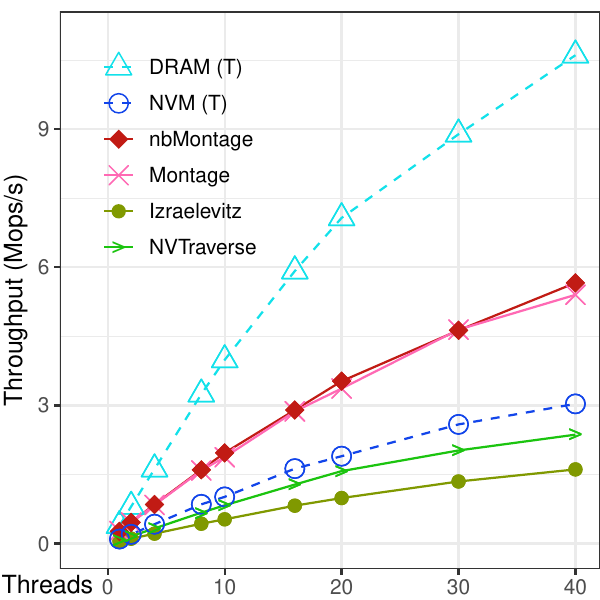}%
        \vspace{-1ex}%
        }
    \hfill
    \subfloat[Skip lists---2:1:1
    get:\lb{}insert:\lb{}remove\label{fig:skiplists_thread}]%
        {\includegraphics[width=.3\textwidth]%
            {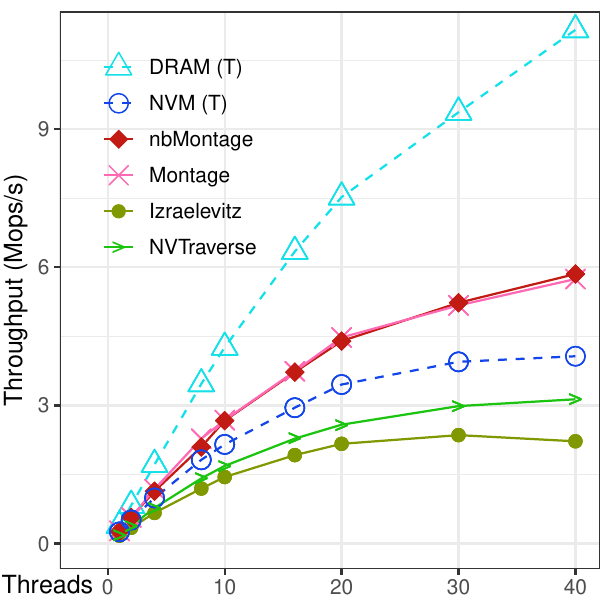}
        \vspace{-1ex}
        }
    \hfill\strut
    \vspace{-1ex}%
    \caption{Throughput of concurrent queues, binary search trees, and skip lists.}
    \vspace{-3ex}%
    \label{fig:throughput-qts}
\end{figure*}

\begin{figure*}
    \centering
        \strut\hfill
        \subfloat[2:1:1
        get:insert:remove\label{fig:hashtables_g50i25r25_thread}]%
            {\includegraphics[width=.48\textwidth]%
                {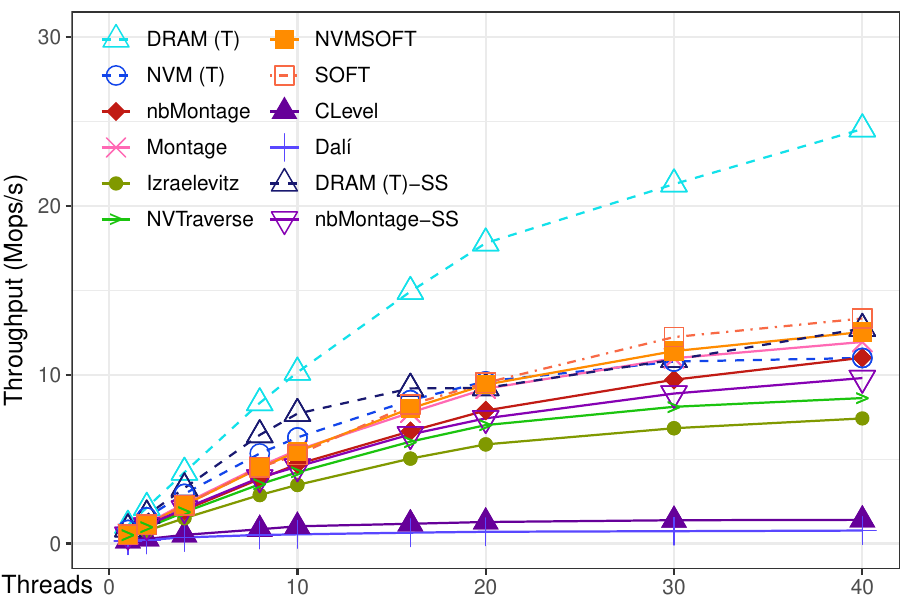}%
            }
        \hfill
        \subfloat[18:1:1
        get:insert:remove\label{fig:hashtables_g90i5r5_thread}]%
            {\includegraphics[width=.48\textwidth]%
                {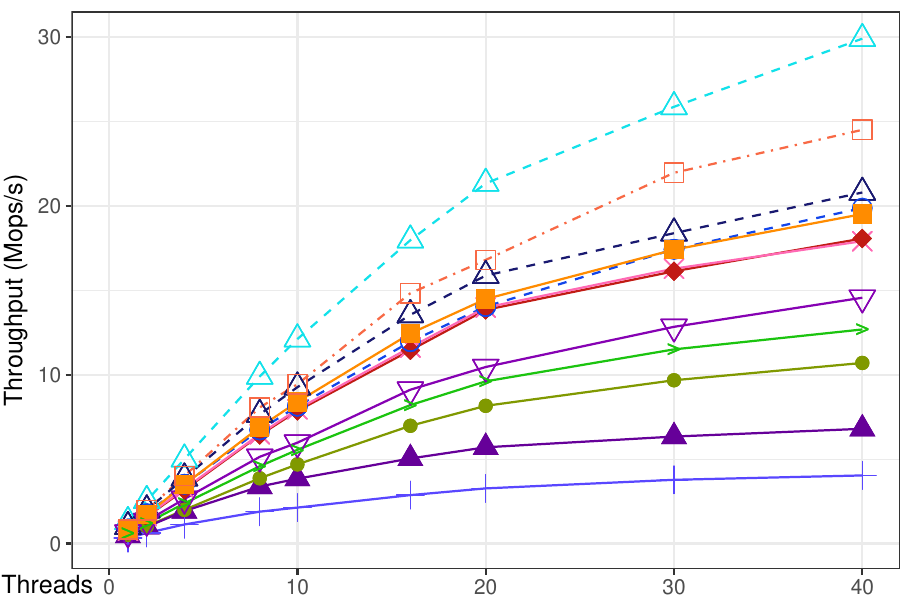}%
            }
        \vspace{-1ex}%
        \caption{Hash table throughput. Options in the left column of
          the key are all variants of Michael's
          nonblocking algorithm.}
        \vspace{-3ex}%
        \label{fig:throughput-ht}
\end{figure*}


\subsection{Overhead of Sync}\label{sec:sync-exp} 

To assess the efficacy of nonblocking epoch advance and of our
mindicator variant, we measured the latency of \sync\ and the throughput
of code that calls \sync\ frequently.  Specifically, using the nbMontage
version of Michael's hash table, on the (2:1:1 get:insert:remove) workload,
we disabled the periodic epoch advance performed by a background thread
and had each worker call \sync\ explicitly.

Average \sync\ latency is shown in Figure~\ref{fig:sync-latency} for
various thread counts and frequencies of calls, on both nbMontage and
its blocking predecessor.  
In all cases, nbMontage completes the typical \sync\ in
1--40\,\textmu s.
In the original Montage, however, \sync\ latency never drops below
5\,\textmu s, and can be as high as 1.3\,ms with high thread count
and low frequency.
 
Hash table throughput as a function of \sync\ frequency is shown in
Figure~\ref{fig:sync-freq} for 40 threads running on Montage and
nbMontage.  For comparison, horizontal lines are shown for various
persistent alternatives (none of which calls \sync). Interestingly,
nbMontage is more than 2$\times$ faster than CLevel even when \sync\
is called in every operation, and starts to outperform NVTraverse once
there are more than about 10 operations per \sync.

\subsection{Recovery}

To assess recovery time, we initialized the nbMontage version of
Michael's hash table with 1--32\,M 1\,KB elements, leading to a total
payload footprint of 1--32\,GB.  With one recovery thread, nbMontage
recovered the 1\,GB data set in 1.4\,s and the 32\,GB data
set in 103.8\,s (22.2\,s to retrieve all in-use blocks; 81.6\,s
to insert them into a fresh hash table).  Eight recovery threads
accomplished the same tasks in 0.3\,s and 17.9\,s.
These times are all within 0.5\,s of recovery times on the original
Montage.

\begin{figure*}
    \centering
    \begin{minipage}[t]{.48\textwidth}
        \centering
        \includegraphics[width=.937\textwidth]%
            {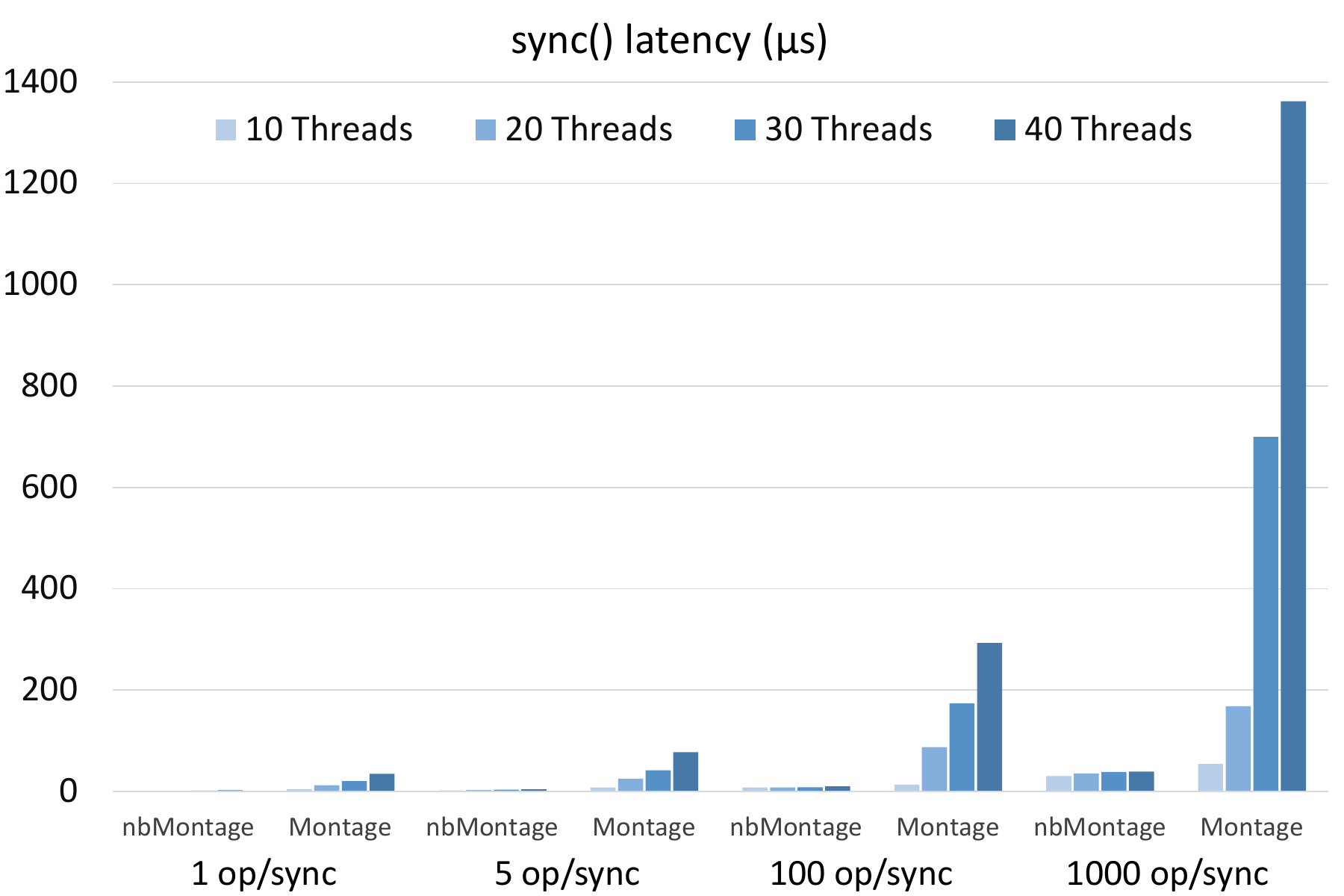}
        \vspace{-1ex}
        \caption{Average latency of \sync\ on hash tables.}
        \label{fig:sync-latency}
    \end{minipage}
    \hfill
    \begin{minipage}[t]{.48\textwidth}
        \centering
        \includegraphics[width=.937\textwidth]%
            {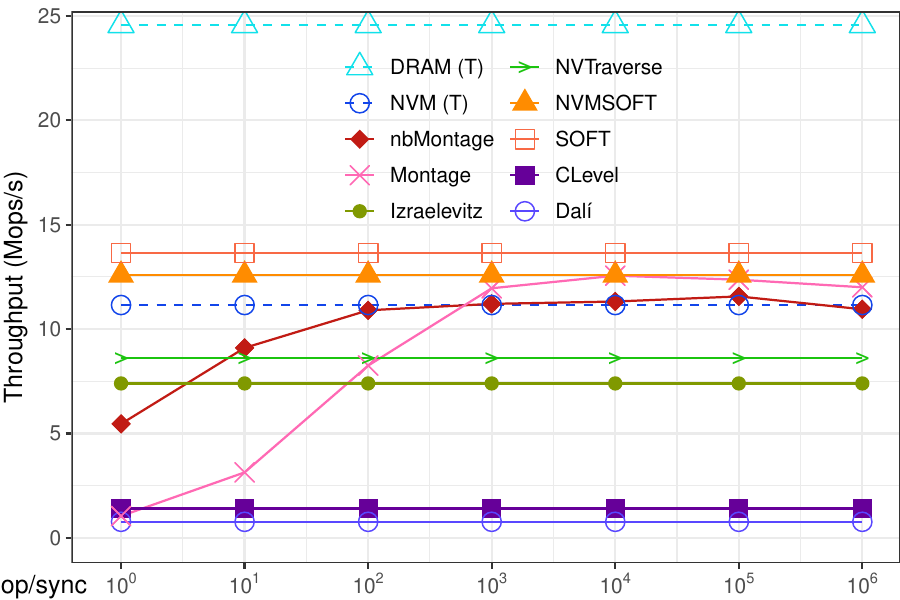}
        \vspace{-1ex}
        \caption{Throughput of hash tables with a \sync\
        every $x$ operations on average.}
        \label{fig:sync-freq}
    \end{minipage}
    \vspace{-3ex}
\end{figure*}

%% file: conclusions.tex
\section{Conclusions}
To the best of our knowledge, nbMontage is the first general-purpose
system to combine buffered durable linearizability with a simple API for
nonblocking data structures and nonblocking progress of the persistence
frontier.  Nonblocking persistence allows nbMontage to provide a very
fast wait-free \sync\ routine and to strictly bound the work that may
be lost on a crash.  Lock-free and wait-free structures, when
implemented on nbMontage, remain lock free; obstruction-free structures
remain obstruction free.

Experience with a variety of nonblocking data structures confirms that
they are easy to port to nbMontage, and perform extremely well---better
in most cases than even hand-built structures that are strictly durably
linearizable.  Given that programmers have long been accustomed to
\sync-ing their updates to file systems and databases, a system with the
performance and formal guarantees of nbMontage appears to be of
significant practical utility.  In ongoing work, we are exploring the
design of a hybrid system that supports both lock-based and nonblocking
structures, with nonblocking persistence in the absence of lock-based
operations.  We also hope to develop a buffered durably linearizable
system for object-based software transactional memory, allowing
persistent operations on multiple data structures to be combined into
failure-atomic transactions.

%% file: appendix.tex
\section{Example nbMontage Data Structure}
\label{app:mmm_hash}

As an example of using the nbMontage API, Fig.~\ref{fig:example}
presents a fragment of Michael's lock-free hash
table~\cite{michael-spaa-2002}, modified for persistence.  Highlighted
parts were changed from the original.

\begin{figure}[h]
\begin{lstlisting}[style=myCustomCppStyle,xleftmargin=12pt,numbersep=4pt,numbers=left]
class MHashTable |\hightlightCode{: public Recoverable}| {
class Payload |\hightlightCode{: public PBlk}| { K key; V val; };
struct Node { // Transient index class
  |\hightlightCode{Payload* payload = nullptr; // Transient-to-persistent pointer}|
  |\hightlightCode{CASObj<Node*> next = nullptr; // Transient-to-transient pointer}|
  Node(K& key, V& val) { |\hightlightCode{payload = pnew<Payload>(key, val);}| }
  ~Node() { |\hightlightCode{if(payload!=nullptr) pdelete(payload);}| }
};
EBRTracker tracker; // Epoch-based memory reclamation
bool find(CASObj<Node*>* &p,Node* &c,Node* &n,K k); // Starting from p, find node >= k and assign to c
void put(K key, V val) { // Insert, or update if the key exists
  Node* new_node = new Node(key, val);
  |\hightlightCode{CASObj<Node*>*}| prev = nullptr;
  Node* curr;
  Node* next;
  tracker.start_op();
  while(true) {
    if (find(prev,curr,next,key)) { // update
      new_node->next.store(curr);
      |\hightlightCode{pdetach(curr->payload);}|
      if(prev->|\hightlightCode{lin\_CAS}|(curr,new_node)) {
        while(!curr->next.CAS(next,mark(next))) next=curr->next.load();
        if(new_node->next.CAS(curr,next)) tracker.retire(curr);
        else find(prev,curr,next,key);
        break;
      }
    } else { // key does not exist; insert
      new_node->next.store(curr);
      if(prev->|\hightlightCode{lin\_CAS}|(curr,new_node))
        break;
    }
  }
  tracker.end_op();
};
\end{lstlisting}
\vspace{-3ex}
\captionsetup{justification=centering}
\caption{Michael's lock-free hash table example (nbMontage-related parts highlighted).}
\label{fig:example}
\vspace{-3ex}
\end{figure}

\iffullver

\section{DCSS}
\label{app:dcss}

Figure~\ref{fig:dcss} provides pseudocode for DCSS-style \lincas.
It supercedes the sketched implementation in Figure~\ref{fig:api}.
CAS-able data is of type \code{CASObj}.  It contains
a 64-bit field that may be either a value or a pointer to a descriptor,
a counter used to avoid ABA problems, and a status field that
differentiates between the value and pointer cases:
\code{INIT} indicates the former; \code{IN\_PROG} the latter.

\begin{figure*}
  \setlength{\columnsep}{0pt}
  \begin{algorithm}[H]
  \scriptsize
  \DontPrintSemicolon 
  \begin{multicols}{2}
  uint2 INIT=0,IN\_PROG=1\;
  uint2 COMMITTED=2,FAILED=3\;
  \Struct{\textup{\texttt{Obj<T>}}}{
    T val\tcp*[l]{64-bit}
    uint62 cnt\;
    uint2 stat\;
  }
  \Struct{\textup{\texttt{CASObj<T>}}}{
    atomic<Obj<T>{}> var\;
  }
  \Struct{\textup{\texttt{Desc}}}{
    uint64 old,new,epoch=0\;
    uint64 type=DESC\;                 
    uint64 tid\_sn=0\;
    atomic<Obj<T>{}> r\_c\_s\;
    \Fn{\textup{\texttt{commit(Obj cur)}}}{\label{code:dcss-commit-begin}
      expected=Obj(cur.val,cur.cnt,IN\_PROG)\;
      desired=Obj(cur.val,cur.cnt,COMMITTED)\;
      r\_c\_s.CAS(expected,desired)\;\label{code:committing-CAS}
    }\label{code:dcss-commit-end}
    \Fn{\textup{\texttt{abort(Obj cur)}}}{
      expected=Obj(cur.val,cur.cnt,IN\_PROG)\;
      desired=Obj(cur.val,cur.cnt,FAILED)\;
      r\_c\_s.CAS(expected,desired)\;
    }
    \Fn{\textup{\texttt{cleanup(Obj cur)}}}{
      tmp=r\_c\_s.load()\;
      \If{\textup{\texttt{tmp.cnt!=cur.cnt or tmp.val!=cur.val\ \ \ \ \ }}}{
        \Return{}
      }
      expected=Obj(this,tmp.cnt,IN\_PROG)\;
      new\_val=tmp.stat==COMMITTED?new:old\;
      desired=Obj(new\_val,tmp.cnt+1,INIT)\;
      CASObj* casobj=tmp.val\;
      casobj$\rightarrow$var.CAS(expected,desired)\;
    }
    \Fn{\textup{\texttt{try\_complete(uint64 addr)}}}{
      cur=r\_c\_s.load()\;
      \If{\textup{\texttt{addr!=cur.val}}}{
        \Return{}
      }
      \If{\textup{\texttt{epoch==global\_epoch.load()}}}{\label{code:dcss-epoch-check}
        commit(cur)\;
      }\Else{
        abort(cur)\;
      }
      cleanup(cur)\;
    }
  }
  \Fn{\textup{\texttt{Obj<T> CASObj<T>::\_load()}}}{
    \Do{\textup{\texttt{r.stat==IN\_PROG}}}{
      r=var.load()\;
      \If{\textup{\texttt{r.stat==IN\_PROG}}}{
        ((Desc*)r.val)$\rightarrow$try\_complete(this)\;\label{code:lin-cas-load-help-complete}
      }
    }
    \Return{\textup{\texttt{r}}}
  }
  \Fn{\textup{\texttt{T CASObj<T>::load()}}}{
    \Return{\textup{\texttt{\_load().val}}}
  }
  \Fn{\textup{\texttt{CASObj<T>::store(T val)}}}{
    \Do{\textup{\texttt{!var.CAS(r,newr)}}}{
      r=var.load()\;
      \If{\textup{\texttt{r.stat==IN\_PROG}}}{
        ((Desc*)r.val)$\rightarrow$try\_complete(this)\;
        r.cnt++\;
        r.stat=INIT\;
      }
      newr=Obj(val,r.cnt+1,INIT)\;
    }
  }
  \Fn{\textup{\texttt{bool CASObj<T>::CAS(T exp,T val)}}}{
    \tcp*[l]{regular CAS without epoch verification}
    r=\_load()\;
    oldr=Obj(exp,r.cnt,INIT)\;
    newr=Obj(val,r.cnt+1,INIT)\;
    \Return{var.CAS(oldr,newr)}\;
  }
  \Fn{\textup{\texttt{bool CASObj<T>::lin\_CAS(T exp,T des)}}}{
    begin\_op()\tcp*[l]{descs[tid].epoch is set here}
    r=\_load()\;\label{code:lin-cas-retry}
    \If{\textup{\texttt{r.val==exp}}}{
      descs[tid].r\_c\_s=Obj(this,r.cnt,IN\_PROG)\;
      descs[tid].old=exp\;
      descs[tid].new=des\;
      newr=Obj(\&descs[tid],r.cnt,IN\_PROG)\;
      \If{\textup{\texttt{var.CAS(r,newr)}}}{\label{code:install-desc}
        descs[tid].try\_complete(this)\;
        \If{\textup{\texttt{COMMITTED==descs[tid].r\_c\_s.stat}}}{
          end\_op()\;
          \Return{true}\;
        }
      }
      \tcp{assume that failed CAS loads new var to r}
      \If{\textup{\texttt{r.val==exp}}}{
        \tcp{failed desc or changed cnt \\
        epoch must have advanced; retry}
        reset\_op()\;
        goto~\ref{code:lin-cas-retry}\;
      }
    }
    abort\_op()\;
    \Return{false}\;
  }
  \end{multicols}
  \vspace{3ex}
  \end{algorithm}
  \caption{DCSS pseudocode, adapted from Harris et al.~\cite{harris-disc-2002}.}
  \label{fig:dcss}
  \vspace{-3ex}
\end{figure*}

Each descriptor includes a similar $\langle$value, counter,
status$\rangle$ field, together with an indication of the old and new
values intended for the corresponding \code{CASObj} and the epoch in
which the change is intended to occur.
In a descriptor, the status field indicates whether an epoch-respecting
CAS is in progress, committed, or failed due to conflict or epoch
advance.
A successful \lincas\ sets up a descriptor, CAS-es a pointer to it into
the to-be-changed \code{CASObj}, double-checks the epoch, CAS-es the
descriptor from \code{IN\_PROG} to \code{COMMITTED}, and then CAS-es the
new value back into the \code{CASObj}.
A conflicting \lincas\ in another thread will try to complete the
descriptor and help update the \code{CASObj} based on the status of
the descriptor.
A successful \lincas\ calls \endop.
A \lincas\ that fails due to conflict calls \abortop.
A \lincas\ that experiences an epoch transition calls \resetop\ and
retries.

Once a thread, $t$, has initialized its \lincas\ descriptor and
installed it in the named \code{CASObj} $j$, arbitrary time may elapse
before $t$ (or another thread) executes the CAS at
line~\ref{code:committing-CAS}.  If that CAS succeeds, the operation $o$
that called \lincas\ can be considered to have linearized at the earlier
\code{load} of \code{global\_epoch} 
at line~\ref{code:dcss-epoch-check}.
If that \code{load} returned $e$ and the CAS succeeds and persists $t$'s
descriptor in epoch $e'$, $o$ will persist when the epoch advances to
$e'+2$, but \emph{it still will have linearized in $e$}---the epoch with
which its payloads were tagged.

It is possible, of course, that a crash will occur before $o$'s
\lincas\ has committed.  Say the crash
occurs in epoch $f \ge e+2$.  If $o$ has not persisted at the crash,
then we know it did not return before the end of epoch $f-2$, implying
that it was pending at the end of epoch~$e$.  Moreover we know
that no other thread will have accessed $j$ ($o$'s \code{CASObj}) before
the end of epoch $f-2$, because it would have completed the CAS at
line~\ref{code:committing-CAS} and $t$'s descriptor should be persisted
by the end of $f-2$.  We are
thus permitted to declare, in recovery, that $o$ did not linearize at
all.  This is true even if $f >\!> e$.


\begin{figure*}
  \setlength{\columnsep}{0pt}
  \begin{algorithm}[H]
  \scriptsize
  \DontPrintSemicolon 
  \SetNoFillComment
  \begin{multicols}{2}
  \Struct{\textup{\texttt{EStat}}}{
    uint63 e\;
    bool status\;
  }
  EStat w[4]\;
  \Fn{\textup{\texttt{end\_op()}}}{
    detaches.clear()\;
    allocs.clear()\;
    \tcc*[l]{new code begins}
    \If{\textup{\texttt{e\_last!=e\_curr}}}{
      \tcp*[l]{first successful op in e\_curr}
      w[e\_curr\%4].CAS($\langle$e\_curr,0$\rangle$,$\langle$e\_curr,1$\rangle$)\;
    }
    \tcc*[l]{new code ends}
    e\_last=e\_curr\;
    e\_curr=0\;
  }
  \Fn{\textup{\texttt{advance()}}}{
    e=global\_epoch.load()\;
    \ForEach{\textup{\texttt{t in mindi whose val==e-1}}}{
      TBP[t][(e-1)\%4].pop\_all()\;
      clwb(descs[t])\;
      update t's val to e in mindi\;
    }
    sfence\;
    \tcc*[l]{new code below}
    \If{\textup{\texttt{w[e\%4]==$\langle$e,0$\rangle$ and w[(e-1)\%4]==$\langle$e-1,0$\rangle$}}}{
      \tcp*[l]{no op succeeds in recent 2 epochs}
      \Return{}
    }
    \If{\textup{\texttt{w[(e+1)\%4]==$\langle$e-3,0$\rangle$}}}{
      \If{\textup{\texttt{w[(e+1)\%4].CAS($\langle$e-3,0$\rangle$,$\langle$e+1,0$\rangle$)}}}{
        goto~\ref{code:advance-epoch}\;
      }
    }
    \If{\textup{\texttt{w[(e+1)\%4]==$\langle$e-3,1$\rangle$}}}{
        w[(e+1)\%4].CAS($\langle$e-3,1$\rangle$,$\langle$e+1,0$\rangle$)\;
    }
    global\_epoch.CAS(e,e+1)\;\label{code:advance-epoch}
  }
  \end{multicols}
  \vspace{3ex}
  \end{algorithm}
  \caption{Pseudocode for \code{advance} and \code{end\_op} with
  logic to skip the epoch change in the absence of a successful
  update operation.}
  \label{fig:epoch-status}
  \vspace{-3ex}
\end{figure*}

\section{Ensuring Progress}
\label{app:progress}

At line~\ref{code:increment-epoch-condition} in
Figure~\ref{fig:pseudo}, to avoid compromising lock freedom,
we assume that an epoch advance happens only when at least one
successful operation has occurred in the two most recent
epochs.  In practice, this will usually be true, given multi-millisecond
epochs and the intended use of \sync.
To enforce such progress, however, the epoch advancer needs to know
whether anything has happened in the two most recent epochs.  To capture
this information, we employ a set of global status words \code{w},
indexed (like TBP and TBF containers) using the low-order bits of the
global epoch.  Each word contains a recent epoch number and a
boolean to indicate progress; each successful
operation tries to CAS \code{w[e\%4]} from \code{$\langle$e,0$\rangle$} to
\code{$\langle$e,1$\rangle$}.
Concurrently, an epoch advance can reuse the status
word of an old epoch by CAS-ing \code{w[e\%4]} from
\code{$\langle$e,\_$\rangle$} to
\code{$\langle$e+4,0$\rangle$}. Pseudocode appears in
Figure~\ref{fig:epoch-status}.

\fi 